\documentclass[a4paper]{article}
\usepackage[latin1]{inputenc} %
\usepackage[T1]{fontenc} %
\usepackage{RR}
\usepackage{hyperref}
\usepackage{amssymb,amsmath}
\usepackage{amsthm}
\usepackage{graphicx}
\RRdate{June 2009}
\newcommand{\R}{\mathbb{R}}
\newcommand{\s}{\mathbb{S}}

\newcommand{\col}{\mathcal{C}}
\newcommand{\F}{\mathcal{F}}
\newcommand{\G}{\mathcal{G}}
\newcommand{\h}{\mathcal{H}}
\newcommand{\N}{\mathcal{N}} % neighborhood
\newcommand{\T}{\mathfrak{T}} % transversals

\newcommand{\lls}{\mathfrak{L}}
\newcommand{\screen}{\mathcal{S}}
\newcommand{\SCREENS}{\mathfrak{S}}
\newcommand{\NORMALS}{\mathfrak{N}}
\newcommand{\NOST}{\mathfrak{X}}
\newcommand{\eps}{\varepsilon}

\newtheorem{theorem}       {Theorem}
\newtheorem{lemma}         {Lemma}
\newtheorem*{unthm}	{Theorem}

\newcommand{\dotp}[1]{\langle#1\rangle}

\RRauthor{
Otfried Cheong
\thanks{\texttt{otfried@kaist.edu}. Dept.~of Computer Science, KAIST, Daejeon, South Korea.}
\and
Xavier Goaoc
\thanks{\texttt{goaoc@loria.fr}. VEGAS Project, INRIA Nancy Grand Est, Nancy, France.}
\and
Andreas Holmsen
\thanks{\texttt{andreash@tclab.kaist.ac.kr}. Dept.~of Computer Science, KAIST, Daejeon, South Korea.}
}
\authorhead{Cheong \& Goaoc \& Holmsen}
\RRtitle{Lower Bounds for Pinning Lines by Balls}
\RRetitle{Lower Bounds for Pinning Lines by Balls}
\titlehead{Lower Bounds for Pinning Lines by Balls}
\RRnote{Otfried Cheong~was supported by the Korea Science and Engineering Foundation Grant~R01-2008-000-11607-0 funded by the Korean government. Andreas Holmsen~was supported by the Brain Korea 21 Project, the School of Information Technology, KAIST, in 2008. The cooperation by Otfried Cheong and Xavier Goaoc was supported by the INRIA \emph{\'Equipe Associ\'ee}~KI.}
\RRresume{Une droite $\ell$ est une \emph{transversale} \`a une famille $\F$ de
convexes de $\R^d$ si elle coupe chaque membre de $\F$. Dans cet
article, nous montrons que pour tout entier $d \geq 3$, il existe une
famille de $2d-1$ boules disjointes de m\^eme rayon dans $\R^d$ sans
droite transversale et telle que toute sous-famille de taille $2d-2$
admet une droite transversale. Cela r\'epond \`a une question de
Danzer de 1957.}

\RRabstract{
A line $\ell$ is a \emph{transversal} to a family $\F$ of convex
objects in~$\R^d$ if it intersects every member of~$\F$. In this paper
we show that for every integer $d\geq 3$ there exists a family of
$2d-1$ pairwise disjoint unit balls in $\R^d$ with the property that
every subfamily of size $2d-2$ admits a transversal, yet any line
misses at least one member of the family. This answers a question of
Danzer from~1957.}
\RRmotcle{G\'eom\'etrie Discr\`ete, Transversales G\'eom\'etriques, Th\'eor\`emes \`a la Helly}
\RRkeyword{Discrete Geometry, Geometric Transversal, Helly-type Theorem}
\RRprojet{Vegas}
\RRdomaine{2} % cas du domaine numero 1
\RCNancy % Nancy - Grand Est
\begin{document}
\RRNo{6961}
\makeRR 

\section{Introduction} 
A straight line that intersects every member of a family $\F$ of compact 
convex sets in $\R^d$ is called a {\em line transversal} to $\F$. An 
important problem in {\em geometric transversal theory} is to give 
sufficient conditions on $\F$ that guarantee the existence of a 
transversal. As an example consider the following result due to Danzer 
\cite{Danzer57}.

\begin{unthm}[Danzer, 1957] A family $\F$ of pairwise disjoint congruent 
disks in the plane has a transversal if and only if every subfamily of 
$\mathcal F$ of size at most 5 has a transversal. \end{unthm}

Simple examples show that the disjointness or congruence can not be 
dropped, nor can the number 5 be reduced. Danzer's theorem has been very 
influential on geometric transversal theory. In 1958, Gr{\"u}nbaum 
\cite{grunbaum} showed that the same result holds when {\em congruent 
disks} is replaced by {\em translates of a square}, and conjectured that 
the result holds also for families of disjoint translates of an arbitrary 
planar convex body. This long-standing conjecture was finally proven by 
Tverberg \cite{tverberg} after partial results were obtained by Katchalski 
\cite{katchalski}.

\begin{unthm}[Tverberg, 1989] A family $\F$ of pairwise disjoint translates 
of a compact convex set in the plane has a transversal if and only if every 
subfamily of $\F$ of size at most 5 has a transversal. \end{unthm}

In a different direction, Danzer's theorem was recetly generalized by the 
present authors together with S.~Petitjean \cite{cghp-hhtdus-08}. This is a 
higher-dimensional analogue of Danzer's theorem, and it solves a problem 
which dates back to Danzer's original article.

\begin{unthm}[Cheong-Goaoc-Holmsen-Petitjean, 2008] A family $\F$ of 
disjoint congruent balls in $\R^d$ has a transversal if and only if every 
subfamily of size at most $4d-1$ has a transversal.\end{unthm}

It should be noted that there are examples which show that Tverberg's 
theorem does not extend to dimensions greater than two \cite{hm-nhtst-04}. 
The theorem just stated provides an upper bound on the {\em Helly-number} 
for line transversals to disjoint congruent balls in $\R^d$. However, a 
missing piece in this particular line of research has been a matching lower 
bound, a problem which again dates back to Danzer's original article. The 
main result of this paper is the following.

\begin{theorem}
  \label{main:lower-bound}
  For every $d \geq 3$, there exists a family of disjoint congruent 
balls in~$\R^d$ which does not have a transversal but where every 
subfamily of size at most $2d-2$ has a transversal. \end{theorem}

Thus the Helly-number for line transversals to disjoint unit balls in 
$\R^d$ is determined up to a factor of 2.

The crucial idea for the proof of Theorem \ref{main:lower-bound} is the 
notion of a {\em pinning}, which was also used in \cite{cghp-hhtdus-08}. 
Intuitively, a line transversal $l$ to a family $\F$ is {\em pinned} if 
every line $l'$ sufficiently close to, but distinct from $l$ fails to be 
transversal to $\F$. In \cite{cghp-hhtdus-08} we showed that if a line is 
pinned by a family $\F$ of disjoint balls then there is a subfamily 
$\G\subset \F$ of size at most $2d-1$ such that $l$ is pinned by $\G$. Here 
we will show that there exists {\em minimal pinning configuration} of 
disjoint (congruent) balls in $\R^d$ of size $2d-1$. By this we mean a 
family of $2d-1$ disjoint balls with a unique transversal $l$ which is 
pinned but where no proper subfamily pins $l$. Theorem 
\ref{main:lower-bound} then follows by slightly shrinking each member of 
the pinning configuration about its center.

There are many surveys that cover geometric transversal theory, among 
others \cite{dgk-htr-63, e-hrctt-93, gpw-gtt-93, w-httgt-04}. For detailed 
information on the transversal properties to 
families of disjoint balls the reader should consult \cite{g-s-08}.

\section{Existence of stable pinnings}

Let $\F$ be a family of compact convex sets in $\R^d$. The set $\T(\F)$ of 
all line transversals to $\F$ forms a
subspace of the affine Grassmanian, which is called the {\em space of
  transversals} to~$\F$. A set $\F$ \emph{pins} (or is a \emph{pinning} of) 
a line~$\ell$ if $\ell$ is an isolated point of~$\T(\F)$; we also say that 
$\ell$ is \emph{pinned} by $\F$, or that the pair $(\F,\ell)$ is a 
\emph{pinning configuration}. If $\F$ pins $\ell$ and no proper subset of 
$\F$ does, then $\F$ is a \emph{minimal} pinning of $\ell$. A minimal 
pinning configuration consisting of pairwise disjoint balls in~$\R^d$ has 
size at most $2d-1$~\cite{bgp-ltdb-08,cghp-hhtdus-08}. Our goal is to show 
that this constant is best possible in all dimensions.

\begin{theorem}
  \label{thm:lower-bound}
  For any $d \geq 2$, there exists a minimal pinning by $2d-1$
  disjoint congruent balls in~$\R^d$.
\end{theorem}

A pinning configuration~$(\F,\ell)$ consisting of disjoint balls
$B_{1},\dots, B_{n}$ in~$\R^{d}$ is \emph{stable} if there exists an
$\eps > 0$ such that any configuration $\F' =\{B'_{1},\dots,
B'_{n}\}$, where the center of $B'_{i}$ has distance at most~$\eps$
from the center of~$B_{i}$ and $B'_{i}$ is tangent to~$\ell$, is also
a pinning of~$\ell$.

\paragraph{Pinning patterns.}

A \emph{halfplane pattern} is a sequence $\h = (H_1, \dots, H_n)$ of
halfplanes in~$\R^{2}$ bounded by lines through the origin.  A
halfplane pattern is a \emph{pinning pattern} if no two halfplanes are
bounded by the same line, and if for every directed line $\ell$ not
meeting the origin and intersecting each halfplane there exist indices
$i < j$ such that $\ell$ exits~$H_j$ before entering~$H_i$.

We first observe that pinning patterns are invariant under small 
perturbations of the halfplanes (that is, if each halfplane is rotated
about the origin by a sufficiently small angle).  More precisely, two
halfplane patterns are equivalent with respect to the pinning pattern
property if the cyclic order of the halfplane boundaries and their
orientation is the same, or, equivalently, if the cyclic order of the
inward and outward normals of the halfplanes is identical.

\begin{figure}[ht]
\begin{center}
\includegraphics{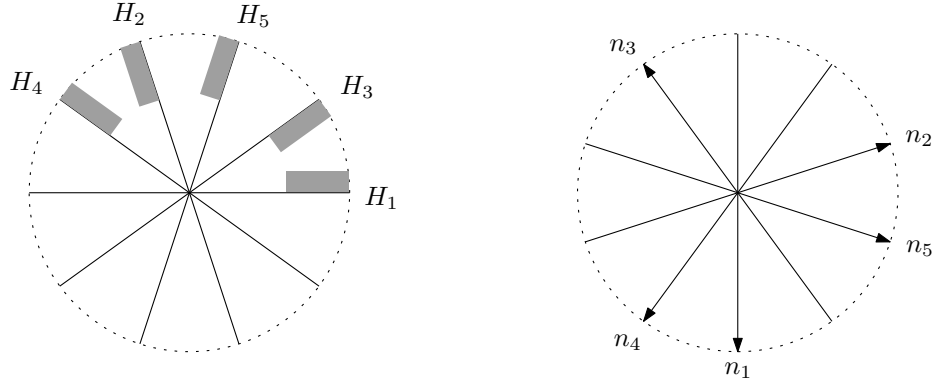}
\caption{\small A $\sigma_5$-patterns, as an arrangement of halfplanes
  through the origin (left) and as a cyclic order of outward and 
  inword normals on $\s^1$ (right).\label{sigma5}}
\end{center}
\end{figure}

Let $n_{i} \in \s^{1}$ denote the outward normal of~$H_{i}$
(throughout the paper, we let $\s^{d-1}$ denote the set of unit
vectors or, equivalently, directions in~$\R^{d}$).  We call a
halfplane pattern of five halfplanes a \emph{$\sigma_{5}$-pattern} if
the outward and inward normals appear in the order (see
Figure~\ref{sigma5})
\[ 
{n_1},-{n_3},{n_5},{n_2},-{n_4},-{n_1},{n_3},-{n_5},-{n_2},{n_4}.
\]
It is easy (but a bit tedious) to verify manually that any
$\sigma_{5}$-pattern is a pinning pattern.  We will give a somewhat
more elegant argument below, but let us first understand the
significance of this fact.

\paragraph{Stable pinnings from pinning patterns.}

The existence of a pinning pattern in the plane allows us to prove the
existence of a stable pinning of a line by five disjoint balls
in~$\R^3$.

Let $\col = (B_{1},\dots,B_{n})$ be a sequence of balls tangent to a
directed line~$\ell$, which touches the balls in the order $B_{1},
\dots, B_{n}$.  We choose a coordinate system where $\ell$ is the
positive $z$-axis.  Projecting ball $B_i$ on the $xy$-plane results in
a disk whose boundary contains the origin; we let $H_i$ denote the
halfplane (in the $xy$-plane) containing this disk and bounded by its
tangent in the origin. We call the halfplane pattern $\h = (H_1,
\dots, H_n)$ the \emph{projection} of~$\col$ along~$\ell$.

\begin{lemma}\label{lem:CP-pinning}
  Let $\col$ be a sequence of disjoint balls in~$\R^{3}$ touching a
  line~$\ell$ in the order of the sequence. If the projection of
  $\col$ along $\ell$ is a pinning pattern, then $\col$ is a pinning
  of~$\ell$.
\end{lemma}
\begin{proof}
  We show that no line other than $\ell$ intersects the members of
  $\col = (B_1, \dots, B_n)$ in the same order, implying that $\ell$
  is pinned by~$\col$. Let $\h = (H_1, \dots, H_n)$ be the projection
  of $\col$, and assume that such a line~$g$ exists. If $g$ is neither
  parallel nor meets $\ell$, its projection $g'$ on the $xy$-plane
  does not go through the origin.  Since $g$ meets each $B_{i}$, $g'$
  intersects each halfplane~$H_{i}$.  Since $\h$ is a pinning pattern,
  there must then be indices $i < j$ such that $g'$ exits $H_j$ before
  entering~$H_i$.  But this implies that $g$ must intersect $B_{j}$
  before~$B_{i}$, a contradiction.

  If $g$ is parallel to $\ell$ then its projection on the $xy$-plane
  is a point lying in $\bigcap_{1 \leq i \leq n} H_i$. Since $\h$ is a
  pinning pattern, this intersection must have empty interior as
  otherwise any line pointing into the sector and not meeting the
  origin does not exit any halfplane; as no two halplanes in $\h$ are
  bounded by the same line, we get $\bigcap_{1 \leq i \leq n}
  H_i=\{O\}$, and so $g$ is $\ell$, a contradiction.

  If $g$ meets $\ell$, we argue that there exists a line that
  intersects the balls in the same order as $\ell$ and is neither
  parallel to nor secant with $\ell$, which brings us back to the
  first case above. Specifically, let $s_i$ be a segment joining a
  point in $B_i \cap \ell$ and a point in $B_i \cap g$ and $g_1$ a
  line through $g \cap \ell$ and the interior of one of the
  $s_i$. Since $g$ and $\ell$ intersect the balls in the same order,
  so does $g_1$. If no $s_i$ is reduced to a single point, $g_1$
  intersects the open balls and can be perturbed into the desired
  line. If some $s_i$ is reduced to a single point, that point is $g
  \cap \ell$ and $g_1$ meets every other segment in its interior; we
  can thus translate $g_1$ to (i) keep intersecting all balls other than
  $B_i$, (ii) move closer to the center of $B_i$ and (iii) stop
  intersecting $\ell$; this yields the desired line.
\end{proof}

In fact, we can strengthen the lemma as follows.
\begin{lemma}\label{lem:CP-stable-pinning}
  Let $\col$ be a sequence of disjoint balls in~$\R^{3}$ touching a
  line~$\ell$ in order of the sequence. If the projection of $\col$
  along $\ell$ is a pinning pattern, then $\col$ is a stable pinning
  of~$\ell$.
\end{lemma}
\begin{proof}
  Consider moving the center of a ball $B_{i}$ in the
  collection~$\col$. In the projection $\h$, the halfplane $H_{i}$
  remains unchanged or rotates about the origin.
  Since we observed above that pinning patterns are invariant under
  sufficiently small rotations of each halfplane, the resulting
  collection is a pinning by Lemma~\ref{lem:CP-pinning}.  And so
  $\col$ is a stable pinning of~$\ell$.
\end{proof}

\paragraph{$\sigma_5$-patterns are pinning patterns.}

The following lemma characterizes pinning patterns.\footnote{The
  necessary condition is actually not used in this paper, and only
  included for completeness}  In addition to proving that
$\sigma_5$-patterns are pinning patterns, we have used a
higher-dimensional version of the sufficient condition to experimentally find
pinning patterns in~$\R^{3}$.
\begin{lemma}
  \label{lem:direction-space-characterization}
  A halfplane pattern $\h$ is a pinning pattern if and only if for any
  direction $u \in \s^1$ there exist indices $i<j<k$ such that $\{n_i,
  n_j, n_k\}$ positively span\footnote{The vectors $\{v_1, \ldots, v_k\}$ positively span the plane if any vector in $\R^2$ can be written as a linear combination of the $v_i$ with non-negative coefficients.} the plane, $\dotp{u,n_i} <0$, and
  $\dotp{u,n_k} >0$.
\end{lemma}
\begin{proof}
  A directed line with direction $u$ exits halfplane $H_i$ if and only if
  $\dotp{u,n_i}>0$ and enters $H_{i}$ if and only if $\dotp{u,n_{i}} <
  0$.
%  \parpic{\includegraphics{Fig/HiHjHk}}
%\noindent
 We first prove that the condition implies that $\h$ is a
pinning pattern. Let $g$ be a line not meeting the origin and meeting
each~$H_{i}$, let $u$ be its direction, and let $i < j < k$ be a
triple satisfying the conditions. Since $\{n_{i}, n_{j}, n_{k}\}$
positively span the plane, we have $H_{i}\cap H_{j}\cap H_{k}=\{0\}$.
%and we have the situation of the figure on the left.
As $g$ does not contain the origin, $g \cap H_j$ and $g \cap (H_i \cap
H_k)$ are disjoint. From $\dotp{u,n_i} <0$ and $\dotp{u,n_k} >0$, we
get that $g$ enters $H_i$ and exits $H_k$. We are thus in one of three
cases: (i) $g$ does not intersect $H_i \cap H_k$, and so exits $H_k$
before entering $H_i$, (ii) $g$ intersects $H_i \cap H_k$ before
$H_j$, and thus exits $H_k$ before entering $H_j$, or (iii) $g$
intersects $H_i \cap H_k$ after $H_j$, and thus exits $H_j$ before
entering $H_i$. In each of these cases, $g$ exits $H_u$ before
entering $H_v$ for some $u<v$. Since this holds for any line $g$ not
containing the origin, it follows that $\h$ is a pinning pattern.

We now prove the other implication. Assume that $\h$ is a pinning
pattern and let $u$ be a direction. We let $g_1$ and $g_2$ be two
lines with direction~${u}$ such that the origin lies in between these
two lines. Since $\h$ is a pinning pattern, there exist indices $a<b$
and $\alpha < \beta$ such that $g_1$ exits $H_{b}$ before
entering~$H_{a}$, and $g_{2}$ exits $H_{\beta}$ before
entering~$H_{\alpha}$. Assume first that no two elements in
$\{a,b,\alpha,\beta\}$ are equal and consider the arrangement of
$\{H_a,H_b,H_\alpha,H_\beta\}$; up to exchanging the roles of $g_1$
and $g_2$, we are in one of the situations (i)--(iii) depicted in
Figure~\ref{3suffice}. In each case, we give an unordered triple of
indices whose halfplanes have outer normals that positively span the
plane (or, equivalently, intersect in exactly the origin): 
\begin{itemize}
\item In situation~(i), the triple is $(a,\alpha,\beta)$ if $a<\alpha$
and $(\alpha,\beta,b)$ otherwise,
\item In
  situation~(ii), the triple is $(a,\beta,b)$ if $a<\beta$ and
  $(\alpha,\beta,b)$ otherwise, 
\item In situation~(iii), the triple
  is $(a,\alpha,\beta)$ if $a<\alpha$ and $(\alpha,a,b)$
  otherwise.
\end{itemize}

\begin{figure}[htb] \begin{center}
 \includegraphics[width=11cm,keepaspectratio]{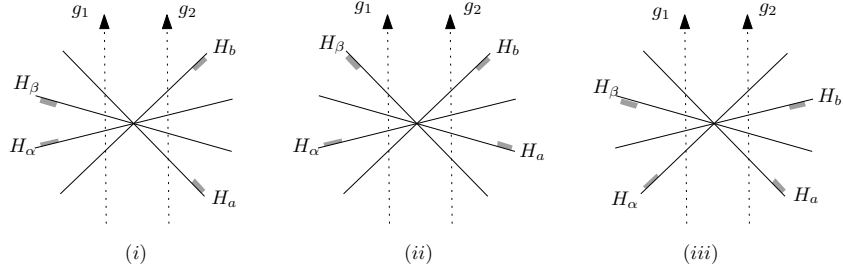} \caption{\small The three possible
 situations for $H_a, H_b, H_\alpha$ and $H_\beta$.\label{3suffice}}
\end{center} \end{figure}

The smallest element must belong to $\{a,\alpha\}$ and the largest to
$\{b,\beta\}$. Since $g_1$ and $g_2$ enters (resp. exit) $H_a$ and
$H_\alpha$ (resp. $H_b$ and $H_\beta$), it follows that $u$ makes a
negative (resp. positive) dot product with the outer normal of the
halfplane with lowest (resp. highest) index; this implies the
condition. 

Consider now the case where $\{a, \alpha, b, \beta\}$ are not all
distinct. Since $g_2$ exits $H_b$ \emph{after} entering $H_a$, at
least three elements of $\{a,b,\alpha,\beta\}$ are pairwise distinct;
for the same reasons as above, this triple of indices satisfies the
condition.
\end{proof}
\begin{lemma}\label{lem:sigma5}
  Any $\sigma_5$-pattern is a pinning pattern.
\end{lemma}
\begin{proof}
  There are four triples of indices of outer normals in a
  $\sigma_5$-pattern that positively span the plane: $\{1,2,3\}$,
  $\{1,3,5\}$, $\{2,3,4\}$ and $\{3,4,5\}$. Figure~\ref{proofsigma5}
  shows, for each triple, the interval of directions that enter the
  first and exit the last member. The union of these (open) intervals
  covers~$\s^1$.  By Lemma~\ref{lem:direction-space-characterization},
  such a pattern is a pinning pattern.
  
\begin{figure}[ht]
    \begin{center}
      \includegraphics{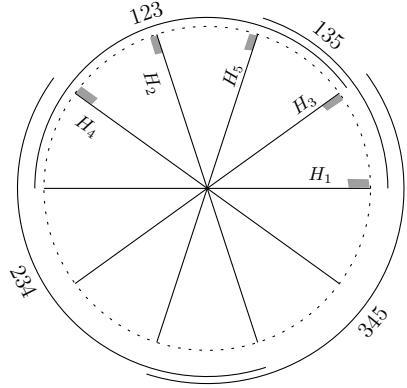}
      \caption{\small A $\sigma_5$-pattern satisfies the condition of
        Lemma~\ref{lem:direction-space-characterization}.\label{proofsigma5}} 
    \end{center}
  \end{figure}
\end{proof}

Combining Lemmas~\ref{lem:CP-stable-pinning} and~\ref{lem:sigma5}
we obtain:
\begin{theorem}\label{lem:exist-stable-3D}
  There exist sequences of five disjoint congruent balls in~$\R^{3}$
  that are stable pinnings.
\end{theorem}

\paragraph{Higher dimensions.}

We now show the existence of stable pinnings by finite families of
disjoint balls in arbitrary dimension.

\begin{theorem}\label{thm:exists-stable-dD}
  For any $d \geq 2$, there exists a stable pinning of a line by
  finitely many disjoint congruent balls in~$\R^d$. 
\end{theorem}
\begin{proof}
  Let $\ell$ be the $x_{d}$-axis in $\R^d$ and let $\Gamma$ be the
  space of all three-dimensional flats containing~$\ell$.  The natural
  homeomorphism between $\Gamma$ and the space of two-dimensional
  linear subspaces of $\R^{d-1}$ implies that $\Gamma$ is compact.

  For every $T \in \Gamma$ we can construct a quintuple $Q_T$ of
  disjoint balls in $\R^d$ tangent to~$\ell$ such that their
  restriction to~$T$ projects along~$\ell$ to a $\sigma_{5}$-pattern.
  By construction, $Q_T$ pins~$\ell$ in~$T$.  By continuity, there
  exists a neighborhood $\N_T$ of $T$ in $\Gamma$ such that $Q_T$ pins
  $\ell$ in any $T' \in \N_T$.  The union of all $\N_T$
  covers~$\Gamma$. Since $\Gamma$ is compact, there exists a finite
  sub-family $\{T_1, \ldots, T_n\}$ such that the union of the
  $\N_{T_i}$ cover~$\Gamma$. Let $\col$ denote the union of
  the~$Q_{T_i}$.

  By construction, $\col$ is a finite collection of balls such that
  the intersection of $\col$ with any 3-flat $T \in \Gamma$ is a
  stable pinning of~$\ell$ in~$T$.  Let $\eps > 0$ be such that any
  collection $\col'$ obtained by perturbing $\col$ by at most~$\eps$
  remains a pinning of $\ell$ in each $T \in \Gamma$. If such a
  perturbation $\col'$ of $\col$ does not pin $\ell$ then there is
  another transversal $\ell'$ of $\col'$ in $\R^{d}$ with the same
  order. There is a three-dimensional affine subspace $T$ containing
  both $\ell$ and~$\ell'$.  Since the set of line transversals with a
  fixed ordering on family of disjoint balls is
  connected~\cite{cghp-hhtdus-08}, this implies that $\ell$ is not
  pinned by $\col' \cap T$ in~$T$, and since $T \in \Gamma$ this is a
  contradiction. Thus, any such perturbation $\col'$ pins $\ell$
  in~$\R^{d}$, implying that $\col$ is a stable pinning.
 
%   Indeed, assume that there is another transversal $\ell'$ of $\col'$
%   in $\R^{d}$ with the same order. There is a three-dimensional affine
%   subspace $T$ containing both $\ell$ and~$\ell'$.  Since the set of
%   line transversals with a fixed ordering on family of disjoint balls
%   is connected~\cite{cghp-hhtdus-08}, this implies that $\ell$ is not
%   pinned by $\col' \cap T$ in~$T$, and since $T \in \Gamma$ this is a
%   contradiction.
 
  Finally, we observe that we can replace a ball $B \in \col$ touching
  $\ell$ in point~$p$ by moving the center of $B$ on the segment
  towards~$p$.  Since this does not change the halfplane pattern in
  the projection, $\col$ remains a stable pinning, and so we can
  choose $\col$ to consist of pairwise disjoint congruent balls.
\end{proof}
By further shrinking the balls, we could even enforce that any two are
separated by a hyperplane orthogonal to~$\ell$.

\section{The size of stable pinnings}

In this section, we will show that families of $k < 2d-1$ balls cannot
be stable pinnings of a line in~$\R^{d}$.  Instead of balls, we will
work with simpler objects we call screens (half-hyperplanes orthogonal
to the line to be pinned), and the lower bound we obtain will carry
over to balls.

\paragraph{Screens and lines.}

Let $\ell$ be the positively oriented $x_{d}$-axis in~$\R^{d}$.
For $\lambda \in \R$ and direction vector $n \in \s^{d-2}$, consider
the set 
\[
\screen(\lambda, n) := \{(x,\lambda) \in \R^{d}
\mid x \in \R^{d-1}, \; \dotp{n,x} \leq 0\},
\]
where the notation $(a,b)$ denotes a vector whose coordinates are
obtained as the concatenation of the coordinates of the vectors~$a$
and~$b$.

We call $\screen(\lambda, n)$ a \emph{screen}. A screen is a
$(d-1)$-dimensional halfspace of a hyperplane orthogonal to~$\ell$;
the screen $\screen(\lambda,n)$ is tangent to~$\ell$ in the
point~$(0,\lambda)$.  We identify $\SCREENS = \R \times \s^{d-2}$ with
the space of all possible screens.

Consider now the space $\lls$ of lines not orthogonal to~$\ell$.  Any
line in~$\lls$ must intersect the planes $x_{d} = 0$ and $x_{d} = 1$.
We identify $\lls$ with $\R^{2d-2}$ by identifying the line meeting
the points $(u_{0}, 0)$ and $(u_{1}, 1)$ with the point $(u_{0},
u_{1}) \in \R^{2d-2}$.

For $(\lambda, n) \in \SCREENS$, let $H(\lambda,n)\subset \lls$ denote
the set of those lines $g \in \lls$ that
intersect~$\screen(\lambda,n)$.
\begin{lemma}
  \label{lem:LLS-halfspace}
  For $(\lambda, n) \in \SCREENS$, the set $H(\lambda, n)$ is the
  halfspace of $\R^{2d-2}$ through the origin with outer normal
  $\Phi(\lambda, n) := ((1-\lambda){n},\lambda{n})$.
\end{lemma}
\begin{proof}
  The line $(u_0,u_1) \in \lls$ intersects the hyperplane $x_{d} =
  \lambda$ in the point $((1-\lambda)u_{0} + \lambda u_{1},
  \lambda)$.  This point lies in 
  $\screen(\lambda, n)$ if and only if
  \[
  \dotp{n,(1-\lambda)u_{0} + \lambda u_{1}} \leq 0,
  \]
  and since
  \[
  \dotp{n,(1-\lambda)u_{0} + \lambda u_{1}} = 
  \dotp{(1-\lambda)n,u_{0}} + \dotp{\lambda n, u_{1}} =
  \dotp{((1-\lambda)n, \lambda n), (u_{0}, u_{1})}
  \]
  the lemma follows.
\end{proof}

Let $\NORMALS \subset \R^{2d-2}$ denote the set of vectors
$\Phi(\lambda, n)$, for some $(\lambda, n) \in \SCREENS$. The function
$\Phi$ is a bicontinuous bijection from $\SCREENS$ to $\NORMALS$, and
so $\SCREENS$ and $\NORMALS$ are homeomorphic. In particular,
$\NORMALS$ is locally homeomorphic to $\R^{d-1}$, and so $\NORMALS$ is
a $(d-1)$-dimensional manifold in~$\R^{2d-2}$.  We need to argue that
it is nowhere contained in a hyperplane, that is, that there is no
neighborhood of a point in $\NORMALS$ that is contained in a
hyperplane.
\begin{lemma}
  \label{lem:linear-independance}
  $\NORMALS$ is nowhere locally contained in a hyperplane of $\lls$.
\end{lemma}
\begin{proof}
  We assume, by way of contradiction, that $\NORMALS$ is contained in
  a hyperplane in a neighborhood of the point $\Phi(\lambda, n) =
  ((1-\lambda)n, \lambda n)$.  Let this hyperplane be
  $\dotp{(a,b),(u_{0},u_{1}))} = c$, where $a,b \in \R^{d-1}$ and $c
  \in \R$.  This means that for $\eps \in \R$ sufficiently small and
  $\eta \in \s^{d-2}$ sufficiently close to~$n$,
  \[
  \dotp{(a,b),((1-\lambda -\eps)\eta, (\lambda+\eps) \eta)} = c.
  \]
  Separating out the terms with $\eps$, we obtain
  \[
  \dotp{(a,b),((1-\lambda)\eta,\lambda\eta} + \eps \dotp{b-a,\eta} = c.
  \]
  Since this holds for any $\eps$ small enough, we must
  have~$\dotp{b-a, \eta} = 0$. Since no neighborhood on $\s^{d-2}$ can
  lie in a hyperplane, it follows that~$b = a$. We thus have
  \[
  \dotp{a,(1-\lambda)\eta + \lambda\eta} = c,
  \]
  which implies $\dotp{a, \eta} = c$.  Again, no neighborhood on
  $\s^{d-2}$ lies in a hyperplane, a contradiction.
\end{proof}

\paragraph{Strict transversals to screens.}

Given a family $\F\subset \SCREENS$ of $k$~screens, a line $g \in
\lls$ is a \emph{strict transversal} of $\F$ if it meets the relative
interior of each screen.  Recall that the line $\ell$ meets every
screen of~$\F$, but since $\ell$ only touches their boundary, it is
not a strict transversal.  If $g \in\lls$ is a strict transversal of
$\F$, then any line $g'\in\lls$ sufficiently close to~$g$ must also be
a strict transversal.  Indeed, $\F$ has a strict transversal if and
only if the intersection of the halfspaces $H(\lambda,n)$ for
$\screen(\lambda,n) \in \F$ has non-empty interior.

\begin{lemma}
  \label{lem:dependent-normals}
  Let $\F$ be a family of $k \leq 2d-2$ screens.  If $\F$ has no
  strict transversal, then the $k$~normals $\Phi(\lambda,n)$, for
  $(\lambda,n) \in \F$, are linearly dependent.
\end{lemma}
\begin{proof}
  If $\F$ has no strict transversal, then the intersection of the
  halfspaces $H(\lambda,n)$ for $(\lambda,n) \in \F$ has empty
  interior.  However, the intersection of $k \leq 2d-2$~halfspaces
  through the origin in $\R^{2d-2}$ can have empty interior only if
  the outward normals of the halfspaces are linearly dependent.  It
  follows that the normals $\Phi(\lambda,n)$, for $(\lambda,n) \in
  \F$, are linearly dependent.
\end{proof}

We can represent\footnote{Note that this is not a bijection as not every point in $\SCREENS^{m}$ represents a set of \emph{distinct} screens.} a family of $m$~screens as a point in $\SCREENS^{m}$.
Let $\NOST_{m} \subset \SCREENS^{m}$ be the space of those families
$\F$ of $m$~screens that have a subfamily $\F' \subseteq \F$ of at
most $2d-2$~screens with \emph{no strict transversal}.
\begin{lemma}
  \label{lem:generic-halfhyperplanes}
  $\NOST_{m}\subset \SCREENS^{m}$ has empty interior.
\end{lemma}
\begin{proof}
Let $\F$ be a family in $\NOST_{m}$. We perturb $\F$, element by
element, into a family $\F'$ with no subset of at most $2d-2$ screens
with linearly dependent vectors. The first element of $\F$ need not be
changed. Assume we already perturbed the first $i$ elements of
$\F$. Every subset of at most $2d-3$ among these $i$ already fixed
normals span a linear subspace of~$\lls$. By
Lemma~\ref{lem:linear-independance}, $\NORMALS$ lies nowhere locally
inside a hyperplane or, since it is a $d-1$-manifold, locally inside a
finite union of hyperplanes. Thus we can choose the $(i+1)^{th}$ element
outside of each of these subspaces, and by induction obtain the
desired perturbation of $\F$.
\end{proof}

\paragraph{A necessary condition for pinning}

Consider now a collection $\col$ of balls tangent to the line~$\ell$
(still assumed to be the $x_{d}$-axis).  If $(p,\lambda)$ is the
center of ball $B \in \col$, we consider the screen $\screen(B) =
\screen(\lambda, -p/||p||)$.  This screen touches $\ell$ in the same
point that $B$ does, and its boundary is contained in the tangent
hyperplane to $B$ at this point.
\begin{lemma}\label{lem:not-pinning}
  Let $\col$ be a collection of balls tangent to~$\ell$.  If the
  family of screens $\F := \{\screen(B)\mid B \in \col\}$ has a strict
  transversal, then $\col$ does not pin~$\ell$.
\end{lemma}
\begin{proof}
  If $\F$ has a strict transversal, then the halfspaces $H(\lambda,n)$
  for $(\lambda,n) \in \F$ intersect with non-empty interior.  This
  implies that there exists a segment $\tau$ in~$\lls$ with one
  endpoint at the origin and which is, except for that point,
  contained in the interior of each halfspace~$H(\lambda,n)$.

  Consider moving a line $g \in \lls$ along~$\tau$.  The trace of $g$
  on the hyperplane $x_{d} = \lambda$ is a straight segment. Since
  $\tau$ lies in the interior of~$H(\lambda,n)$, this trace lies in
  the relative interior of~$\screen(\lambda,n)$.  But this implies
  that if we make $\tau$ sufficiently short, the trace also lies in
  the interior of each ball~$B$.  It is therefore possible to move a
  line $g$, starting with $g = \ell$, while intersecting each
  ball~$B$.  It follows that $\col$ is not a pinning of~$\ell$.
\end{proof}
We now obtain the desired lower bound on the size of stable pinning
configurations of balls:
\begin{theorem}\label{theo:instable}
  Any pinning of a line by $k \leq 2d-2$ balls in $\R^d$ is instable.
\end{theorem}
\begin{proof}
  Let $\col$ be a pinning of~$\ell$. Let $\F := \{\screen(B) \mid B
  \in \col\}$ be the corresponding family of screens.  By
  Lemma~\ref{lem:not-pinning}, $\F$ does not have a strict
  transversal, and so $\F \in \NOST_{k}$.  By
  Lemma~\ref{lem:generic-halfhyperplanes}, $\NOST_{k}$ has empty
  interior, and so we can find $\F' \in\SCREENS^{k} \setminus
  \NOST_{k}$ arbitrarily close to~$\F$.  Since $\F'$ can be realized
  as the set of screens of a perturbation of~$\col$, the theorem
  follows.
\end{proof}

\section{Consequences}

\paragraph{Lower bound for minimal pinnings.}
Theorem~\ref{thm:lower-bound} follows immediately from 
Theorem~\ref{thm:exists-stable-dD} and the following lemma:
\begin{lemma}\label{lem:exists-stable-minimal}
  If $\col$ is a stable pinning of a line by finitely many balls in
  $\R^d$, then there exist minimal pinnings of $\ell$ by $2d-1$ balls
  arbitrarily close to some subset of $\col$ of size~$2d-1$.
\end{lemma}
\begin{proof}
  Let $m$ denote the number of balls in $\col$, and let $\F :=
  \{\screen(B) \mid B \in \col\}$ be the corresponding family of
  screens.  Since $\NOST_{m} \subset \SCREENS^{m}$ has empty interior,
  we can find a family $\F' \subset \SCREENS^{m} \setminus \NOST_{m}$
  arbitrarily close to~$\F$.  Let $\col'$ be the correspondingly
  perturbed family of balls.

  By definition of $\NOST_{m}$ and Lemma~\ref{lem:not-pinning}, no
  subfamily of at most $2d-2$ balls of $\col'$ is a pinning.  However,
  any minimal pinning of a line by disjoint balls in~$\R^{d}$ has size
  at most~$2d-1$~\cite{cghp-hhtdus-08}, so there must be a subfamily
  $\col'' \subset \col'$ of $2d-1$ balls that is minimally pinning.
\end{proof}

\paragraph{Helly number for transversals to disjoint unit balls.} 
Hadwiger's transversal theorem \cite{h-uegt-57} can be extended to families 
of disjoint balls in arbitrary dimension~\cite{bgp-ltdb-08,cghp-hhtdus-08}: 
a family $\F$ of disjoint balls in $\R^d$ has a line transversal if and 
only if there is an ordering on $\F$ such that every $h_d$ members have a 
line transversal consistent with that ordering. The smallest such constant 
$h_d$ is at most $2d$ and at least the size of the largest minimal pinning 
family of disjoint balls in $\R^d$; Theorem~\ref{thm:lower-bound} implies 
that this number is $2d-1$ or $2d$. Similarly, we obtain that $2d-1$ is a 
lower bound for the Helly number of the generalization of Helly's theorem 
to sets of transversals to disjoint (unit) balls.

\begin{proof}[Proof of Theorem \ref{main:lower-bound}]
The proof of Theorem~\ref{thm:exists-stable-dD} shows that there
exists a stable pinning of a line $\ell$ by finitely many disjoint
unit balls in $\R^d$ such that any two balls can be separated by a
hyperplane orthogonal to $\ell$. Lemma~\ref{lem:exists-stable-minimal}
now implies that there exists a minimal pinning $\F$ of a line $\ell$
by $2d-1$ disjoint unit balls in $\R^d$ such that any line
intersecting a subset of the balls does so in an order consistent with
the geometric permutation induced by $\ell$. Since $\F$ is a pinning
of $\ell$ by disjoint balls, $\F$ has no other transversal consistent
with the geometric permutation induced by $\ell$. The statement then
follows in the case of open balls. Reducing the radii of the balls
slightly gives a similar construction for closed balls.
\end{proof}

\section{Final remarks}

Our lower bound construction is hardly
effective, given its use of the compactness of the set of $3$-spaces
through a fixed line; actually constructing minimal pinnings of size
$2d-1$ (or any given size) seems challenging. Another natural question
is whether any of the results obtained for pinning lines by disjoint
balls extend to more general pinnings, for instance pinnings of lines
by (disjoint) convex sets. In that direction, little is known, not
even whether the size of minimal pinnings is bounded by a function
of~$d$.

%\newpage
\bibliographystyle{plain}
\bibliography{trans}

\end{document}